\documentclass[11pt]{article}

\def\showauthornotes{0}

\def\showkeys{0}
\def\showdraftbox{0}
\def\showcolorlinks{1}
\def\usemicrotype{1}
\def\showfixme{0}

\usepackage{etex}

\usepackage[l2tabu, orthodox]{nag}

\usepackage{xspace,enumerate}

\usepackage[dvipsnames]{xcolor}

\usepackage[american]{babel}

\usepackage{mathtools}

\usepackage{amsthm}

\newtheorem{theorem}{Theorem}[section]
\newtheorem*{theorem*}{Theorem}

\newtheorem{proposition}[theorem]{Proposition}
\newtheorem*{proposition*}{Proposition}
\newtheorem{lemma}[theorem]{Lemma}
\newtheorem*{lemma*}{Lemma}
\newtheorem{corollary}[theorem]{Corollary}

\newtheorem*{conjecture*}{Conjecture}

\newtheorem*{fact*}{Fact}

\newtheorem*{hypothesis*}{Hypothesis}

\theoremstyle{definition}
\newtheorem{definition}[theorem]{Definition}
\newtheorem*{definition*}{Definition}

\newtheorem{algorithm}[theorem]{Algorithm}

\newtheorem*{problem*}{Problem}

\theoremstyle{remark}
\newtheorem{claim}[theorem]{Claim}
\newtheorem*{claim*}{Claim}

\newtheorem*{remark*}{Remark}

\newtheorem*{observation*}{Observation}

\usepackage[
letterpaper,
top=1in,
bottom=0.9in,
left=1.0in,
right=1.0in]{geometry}

\usepackage{amsmath,amsfonts,amssymb}

\ifnum\showkeys=1
\usepackage[color]{showkeys}
\fi

\ifnum\showcolorlinks=1
\usepackage[
pagebackref,
colorlinks=true,
urlcolor=blue,
linkcolor=blue,
citecolor=OliveGreen,
]{hyperref}
\fi

\ifnum\showcolorlinks=0
\usepackage[
pagebackref,
colorlinks=false,
pdfborder={0 0 0}
]{hyperref}
\fi

\usepackage{prettyref}
\newcommand{\pref}{\prettyref}

\newcommand{\savehyperref}[2]{\texorpdfstring{\hyperref[#1]{#2}}{#2}}

\newrefformat{eq}{\savehyperref{#1}{\textup{(\ref*{#1})}}}
\newrefformat{eqn}{\savehyperref{#1}{\textup{(\ref*{#1})}}}
\newrefformat{lem}{\savehyperref{#1}{Lemma~\ref*{#1}}}
\newrefformat{def}{\savehyperref{#1}{Definition~\ref*{#1}}}
\newrefformat{thm}{\savehyperref{#1}{Theorem~\ref*{#1}}}
\newrefformat{cor}{\savehyperref{#1}{Corollary~\ref*{#1}}}
\newrefformat{cha}{\savehyperref{#1}{Chapter~\ref*{#1}}}
\newrefformat{sec}{\savehyperref{#1}{Section~\ref*{#1}}}
\newrefformat{app}{\savehyperref{#1}{Appendix~\ref*{#1}}}
\newrefformat{tab}{\savehyperref{#1}{Table~\ref*{#1}}}
\newrefformat{fig}{\savehyperref{#1}{Figure~\ref*{#1}}}
\newrefformat{hyp}{\savehyperref{#1}{Hypothesis~\ref*{#1}}}
\newrefformat{alg}{\savehyperref{#1}{Algorithm~\ref*{#1}}}
\newrefformat{rem}{\savehyperref{#1}{Remark~\ref*{#1}}}
\newrefformat{item}{\savehyperref{#1}{Item~\ref*{#1}}}
\newrefformat{step}{\savehyperref{#1}{Step~\ref*{#1}}}
\newrefformat{conj}{\savehyperref{#1}{Conjecture~\ref*{#1}}}
\newrefformat{fact}{\savehyperref{#1}{Fact~\ref*{#1}}}
\newrefformat{prop}{\savehyperref{#1}{Proposition~\ref*{#1}}}
\newrefformat{prob}{\savehyperref{#1}{Problem~\ref*{#1}}}
\newrefformat{claim}{\savehyperref{#1}{Claim~\ref*{#1}}}
\newrefformat{relax}{\savehyperref{#1}{Relaxation~\ref*{#1}}}
\newrefformat{red}{\savehyperref{#1}{Reduction~\ref*{#1}}}
\newrefformat{part}{\savehyperref{#1}{Part~\ref*{#1}}}
\newrefformat{obs}{\savehyperref{#1}{Observation~\ref*{#1}}}
\newrefformat{corr}{\savehyperref{#1}{Corollary~\ref*{#1}}}

\newcommand{\Sref}[1]{\hyperref[#1]{\S\ref*{#1}}}

\usepackage{nicefrac}
\ifnum\usemicrotype=1
\usepackage{microtype}
\fi

\ifnum\showauthornotes=1
\newcommand{\Authornote}[2]{{\sffamily\small\color{red}{[#1: #2]}}}
\newcommand{\Authornotecolored}[3]{{\sffamily\small\color{#1}{[#2: #3]}}}
\newcommand{\Authorcomment}[2]{{\sffamily\small\color{gray}{[#1: #2]}}}
\newcommand{\Authorstartcomment}[1]{\sffamily\small\color{gray}[#1: }

\newcommand{\Authorfnote}[2]{\footnote{\color{red}{#1: #2}}}
\newcommand{\Authorfixme}[1]{\Authornote{#1}{\textbf{??}}}
\newcommand{\Authormarginmark}[1]{\marginpar{\textcolor{red}{\fbox{\Large #1:!}}}}
\else
\newcommand{\Authornote}[2]{}
\newcommand{\Authornotecolored}[3]{}
\newcommand{\Authorcomment}[2]{}
\newcommand{\Authorstartcomment}[1]{}

\newcommand{\Authorfnote}[2]{}
\newcommand{\Authorfixme}[1]{}
\newcommand{\Authormarginmark}[1]{}
\fi

\ifnum\showfixme=0

\fi

\usepackage{boxedminipage}

\newcommand{\Esymb}{\mathbb{E}}
\newcommand{\Psymb}{\mathbb{P}}

\DeclareMathOperator*{\E}{\Esymb}

\DeclareMathOperator*{\ProbOp}{\Psymb}

\renewcommand{\Pr}{\ProbOp}

\newcommand{\textparen}[1]{\text{(#1)}}

\ifx\because\undefined
\newcommand{\because}[1]{\textparen{because #1}}
\else
\renewcommand{\because}[1]{\textparen{because #1}}
\fi

\newcommand{\mper}{\,.}

\newcommand\bdot\bullet

\ifx\mathds\undefined % use double stroke fonts if available

\else
}
\fi

\DeclareMathOperator{\polylog}{polylog}

\newcommand{\cC}{\mathcal C}

\ifnum\showdraftbox=1
\newcommand{\draftbox}{\begin{center}
  \fbox{%
    \begin{minipage}{2in}%
      \begin{center}%
          \Large\textsc{Working Draft}\\%
        Please do not distribute%
      \end{center}%
    \end{minipage}%
  }%
\end{center}
\vspace{0.2cm}}
\else
\newcommand{\draftbox}{}
\fi

\let\epsilon=\varepsilon

\numberwithin{equation}{section}

\newcommand\MYcurrentlabel{xxx}

\newcommand{\MYstore}[2]{%
  \global\expandafter \def \csname MYMEMORY #1 \endcsname{#2}%
}

\newcommand{\MYload}[1]{%
  \csname MYMEMORY #1 \endcsname%
}

\newcommand{\MYnewlabel}[1]{%
  \renewcommand\MYcurrentlabel{#1}%
  \MYoldlabel{#1}%
}

\newcommand{\MYdummylabel}[1]{}

\newcommand{\torestate}[1]{%
  \let\MYoldlabel\label%
  \let\label\MYnewlabel%
  #1%
  \MYstore{\MYcurrentlabel}{#1}%
  \let\label\MYoldlabel%
}

\newcommand{\restatetheorem}[1]{%
  \let\MYoldlabel\label
  \let\label\MYdummylabel
  \begin{theorem*}[Restatement of \prettyref{#1}]
    \MYload{#1}
  \end{theorem*}
  \let\label\MYoldlabel
}

\newcommand{\restatedef}[1]{%
  \let\MYoldlabel\label
  \let\label\MYdummylabel
  \begin{definition*}[Restatement of \prettyref{#1}]
    \MYload{#1}
  \end{definition*}
  \let\label\MYoldlabel
}

\newcommand{\restatelemma}[1]{%
  \let\MYoldlabel\label
  \let\label\MYdummylabel
  \begin{lemma*}[Restatement of \prettyref{#1}]
    \MYload{#1}
  \end{lemma*}
  \let\label\MYoldlabel
}

\newcommand{\restateprop}[1]{%
  \let\MYoldlabel\label
  \let\label\MYdummylabel
  \begin{proposition*}[Restatement of \prettyref{#1}]
    \MYload{#1}
  \end{proposition*}
  \let\label\MYoldlabel
}

\newcommand{\restatefact}[1]{%
  \let\MYoldlabel\label
  \let\label\MYdummylabel
  \begin{fact*}[Restatement of \prettyref{#1}]
    \MYload{#1}
  \end{fact*}
  \let\label\MYoldlabel
}

\newcommand{\restateobs}[1]{%
  \let\MYoldlabel\label
  \let\label\MYdummylabel
  \begin{observation*}[Restatement of \prettyref{#1}]
    \MYload{#1}
  \end{observation*}
  \let\label\MYoldlabel
}
\newcommand{\restate}[1]{%
  \let\MYoldlabel\label
  \let\label\MYdummylabel
  \MYload{#1}
  \let\label\MYoldlabel
}

\newcommand{\addreferencesection}{
  \phantomsection
  \addcontentsline{toc}{section}{References}
}

\newcommand{\eps}{\epsilon}

\let\origparagraph\paragraph
\renewcommand{\paragraph}[1]{\origparagraph{#1.}}

\allowdisplaybreaks

\usepackage[newenum,newitem,increaseonly]{paralist}
\usepackage{enumitem}

\usepackage{comment}

\let\citet\cite

\theoremstyle{definition}

\DeclareUrlCommand\email{}

\newcommand{\restateproblem}[2]{%
  \let\MYoldlabel\label
  \let\label\MYdummylabel
  \begin{problem*}[Restatement of \prettyref{#1}, {#2}]
    \MYload{#1}
    \end{problem*}
  \let\label\MYoldlabel
}

\newcommand{\tO}{\widetilde{O}}
\newcommand{\tOmega}{\widetilde{\Omega}}
\newcommand{\tTheta}{\widetilde{\Theta}}

\newcommand{\ignore}[1]{}%

\title{Computing exact minimum cuts without knowing the graph}

\author{%
\normalsize
Aviad Rubinstein\thanks{UC Berkeley,
  \protect\email{aviad@eecs.berkeley.edu}. This research was supported by Microsoft Research PhD Fellowship. It was also supported in part by NSF grant CCF1408635 and by Templeton Foundation grant 3966. This work was done in part at the Simons Institute for the Theory of Computing.}
\and
\normalsize
Tselil Schramm\thanks{UC Berkeley, \protect\email{tscrhamm@cs.berkeley.edu}. Supported by an NSF Graduate Research Fellowship (1106400).}
\and
\normalsize
Matt Weinberg\thanks{Princeton University, \protect\email{smweinberg@princeton.edu}. Work done in part while the author was visiting the Simons Instutute for the Theory of Computing.}
}

\date{}

\begin{document}

\maketitle

\draftbox

\thispagestyle{empty}

\begin{abstract}
% !TeX root = ..\writeup.tex

We give query-efficient algorithms for the global min-cut and the s-t cut problem in unweighted, undirected graphs.
Our oracle model is inspired by the submodular function minimization problem:
on query $S \subset V$, the oracle returns the size of the cut between $S$ and $V \setminus S$.

We provide algorithms computing an exact minimum $s$-$t$ cut in $G$ with $\tO(n^{5/3})$ queries, and computing an exact global minimum cut of $G$ with only $\tO(n)$ queries (while learning the graph requires $\tTheta(n^2)$ queries).

\end{abstract}
\clearpage
\thispagestyle{empty}

\clearpage
\setcounter{page}{1}

% !TeX root = ..\writeup.tex
\section{Introduction}
We give new algorithms for the minimum cut and $s$-$t$ minimum cut problems in an unweighted, undirected graph $G = (V,E)$.
Our algorithms do not assume access to the entire graph $G$;
rather, they only interact with an oracle that, on query $S$, returns the value $c(S)$ of the cut between $S$ and $V \setminus S$.
Our goal is to minimize the number of queries to the oracle while computing (exact) optimum cuts.

\subsubsection*{Three easy algorithms}

How many queries should we expect to be necessary?
It is not hard to see that $\binom{n}{2} +n  = O(n^2)$ queries suffice:%
\footnote{We follow the standard convention and use $n = |V|$ to denote the number of vertices and $m = |E|$ to denote the number of edges. We also use $\tO(x)$ to denote $O(x \cdot \polylog(x))$, and similarly for $\tilde \Omega(\cdot)$ and $\tilde \Theta(\cdot)$.}
To find out whether there is an edge between $u$ and $v$, we can query the oracle for $\{u\}$,  $\{v\}$, and $\{u,v\}$. The edge is present iff $c(\{u\})+c(\{v\})-c(\{u,v\}) > 0$. After querying all $n$ singletons and $\binom{n}{2}$ pairs, we have learned the entire graph.
In fact, we can improve slightly: the results of the cut queries are linear in $\binom{n}{2}$ unknown variables; thus $\binom{n}{2}$ linearly independent queries suffice.

For sparse graphs we can do even better. One can find a neighbor of a non-isolated vertex $v$ in $O(\log n)$ queries using a ``lion in the desert'' algorithm (see \pref{lem:1edge}). Because of this, we can learn the entire graph using $\tO(n+m)$ queries.
But for dense graphs, $\widetilde \Omega(n^2)$ queries are necessary to learn the entire graph by a simple information theoretic argument (each query reveals at most $O(\log n)$ bits).

\subsubsection*{The search for a lower bound}

Because $\widetilde\Omega(n^2)$ queries are needed to learn the graph, it is natural to conjecture that $\widetilde \Theta(n^2)$ is the optimal query complexity for computing the min cut. Such a lower bound would be of considerable interest: after breakthrough progress in recent years, $\tO(n^2)$ is also the state of the art query complexity for the more general problem of submodular function minimization over subsets of $n$ items \cite{LSW15-submodular_min}.\footnote{Note that the more recent work of~\cite{CLSW16-submodular_pseudo} requires $\tilde{O}(nM^3)$ queries, where the function is integral and $M$ is the maximum value the function takes; for cuts in graphs, $M = n^2/4$!}
Recent work of~\cite{CLSW16-submodular_pseudo} indeed rules out certain kinds of algorithms with subquadratic query complexity, but determining whether submodular function minimization requires $\widetilde \Omega(n^2)$ queries remains an exciting open problem (see Section~\ref{sec:relatedwork} for more discussion), and graph cuts seemed like a promising candidate.

\subsubsection*{Our results}

In defiance of our intuition, we provide algorithms for global minimum cut and minimum $s$-$t$ cut that use a truly subquadratic number of queries. Our main results are:
\begin{theorem}[Global Min Cut]\label{thm:gmc}
There exists a randomized algorithm that with high probability computes an exact global minimum cut in simple graphs using $\tO(n)$ queries.
\end{theorem}
\begin{theorem}[Min $s$-$t$ Cut]\label{thm:stc}
There exists a randomized algorithm that with high probability computes an exact $s$-$t$ minimum cut in simple graphs using $\tO(n^{5/3})$ queries.
\end{theorem}

It is worth mentioning that while our focus is query efficiency, all our algorithms run in polynomial time ($\tO(n^2)$ or faster).

\subsection*{Techniques}
All our algorithms are quite simple.
Both results can be obtained using the following ``meta-algorithm'':
 (1) subsample a subquadratic number of edges;
 (2) compress the (original) graph by contracting all ``safe'' edges (i.e. those that do not cross the optimum cut, with high confidence, based on the subsample); and
 (3) learn all remaining edges.

It is not hard to see that uniform sampling does not work for Step (1). For example, for the $s$-$t$ minimum cut problem, consider a giant clique that is disconnected from the component containing $s$ and $t$ --- sampling the edges from the clique with the same probability as those from the ``real'' graph is clearly a bad idea.

Instead of uniform sampling, we build on the {\em edge strength}-based sampling due to Bencz\'ur and Karger~\cite{BK15-strength_sampling}, which in $\tO(m)$ queries yields graphs with few edges that approximate every cut well.
Calculating the edge-strengths with $o(m)$ queries is non-trivial.
Instead of computing the strength of every edge, we sub-sample the graph at different resolutions to classify the vertices of the graph into strongly connected components. See \pref{sec:bksamp} for details.

For the global minimum cut problem we also provide an even simpler algorithm, which avoids edge-strength sampling:
in a preprocessing step, we contract edges uniformly at random, as in Karger's algorithm \cite{Karger93-Kargers_Algorithm}.
After the right number of edge contractions, it suffices to sample edges uniformly at random in Step (1).

\subsection{Related work}\label{sec:relatedwork}
Graph cut minimization is a classical algorithms topic, with work dating back to Ford and Fulkerson \cite{FF62-Ford_Fulkerson}, and too many consequent results to list.
Of particular relevance to the present paper are the works of Karger and co-authors, including \cite{Karger93-Kargers_Algorithm,Karger99-skeletons,KS96-Karger_Stein,BK15-strength_sampling,KL15-st_cuts}, which give randomized algorithms for computing minimum cuts and related quantities efficiently---in particular, this line of work establishes methodology for randomly compressing graphs while preserving cut information, which has been used in numerous follow-up works. Though our goal is \emph{query} efficiency rather than runtime efficiency, we very much rely on their insights.

Another work of note is the recent result of Kawarabayashi and Thorup \cite{KT15-deterministic_min_cut}, who show that the global minimum cut can be computed deterministically in $\tO(m)$ time.
Though their setting differs from ours, our works are similar in that we too require structural theorems about the number of edges participating in minimum cuts in the graph (e.g. Lemma~\ref{lem:cover}).

As mentioned above, our initial motivation for studying the min cut problem in this oracle model came from submodular function minimization (SFM).
SFM was first studied by Gr\"{o}tschel, Lov\'asz, and Schrijver in the 1980's \cite{GlLS81-submodular_min}, and has since been a popular topic of study (see e.g. \cite{Fujishige05-textbook} for a thorough treatment).
For a submodular function over $n$ items, the current best general algorithm requires $\tO(n^2)$ oracle queries \cite{LSW15-submodular_min}.
Furthermore,~\cite{CLSW16-submodular_pseudo} suggest that $\tTheta(n^2)$ is indeed the right bound: they prove an $\Omega(n^2)$ lower bound on the number of oracle calls made by a restricted class of algorithms
(those that access the submodular function by naively evaluating the subgradient of its Lov\'asz extension).
For general algorithms, the current best lower bound is $\Omega(n)$ queries, due to Harvey \cite{Harvey08-matroid_intersection} building on the work of Hajnal, Mass and Turan \cite{HMT88}.
Graph cuts and $s$-$t$ cuts are canonical examples of symmetric and asymmetric submodular functions,\footnote{A submodular function is symmetric if $f(S) = f(\bar{S})$ for all $S$. $\emptyset$ and $[n]$ are always minimizers of a symmetric submodular function, so the ``symmetric submodular function minimization problem'' is to find a non-trivial minimizer (i.e. the minimizer $\notin \{\emptyset,[n]\}$), of which global min cut is a special case.} and while it would be natural to conjecture that $\tOmega(n^2)$ queries are needed for graph cut problems, our work demonstrates that these problems do not provide a lower bound matching \cite{LSW15-submodular_min}'s algorithm (at least in unweighted graphs, and for randomized algorithms).

Note that there are works that bypass the $\tTheta(n^2)$ oracle queries barrier for special cases of SFM.
For example,~\cite{CLSW16-submodular_pseudo} provide an algorithm with $\tTheta(nM^3)$ oracle queries when the function value is integral and bounded within $[1,M]$ (for min cut $M$ may be as large as $\Theta(n^2)$). Another special case of interest are decomposable submodular functions (e.g.~\cite{SK10-decomposable, % Kolmogorov12-decomposable, FJPZ13-decomposable,EN15-decomposable,
NJJ14-decomposable}).

We also mention a sequence of papers \cite{CK08,M10,BM11} which study the query efficiency of \emph{learning} a graph under a similar query model (in which each cut query can be implemented in $O(1)$ queries).
This series of papers establishes that $\Theta(m\log (n^2/m)/\log m)$ queries are to learn a graph on $n$ vertices and $m$ edges.
This improves upon our naive algorithm for learning a graph (see \pref{lem:1edge}) by polylogarithmic factors.
Graph reconstruction has also been considered with different oracles, e.g. a distance oracle~\cite{KKU95-distance_oracle, MZ13-distance_oracle, KMZ15-distance_oracles, KKL16-distance_oracle}.

To our knowledge, no other works have previously considered the query complexity of graph cuts.
However, the task of compressing graph cut information into efficient structures has been studied before from a variety of angles:  sketching \cite{ACKQWZ16-sketches,KK15-sketches}, spectral sparsifiers \cite{BSS14-sparseifier}, streaming spectral sparsifiers \cite{KapralovLMMS14}, skeletons \cite{Karger99-skeletons,BK15-strength_sampling},  backbones \cite{CEPP12-backbones}, and cactus representations \cite{DKL76-cactus, NV91-cactus},
%KT86-cactus,
%KP09-cactus_in_O(m) }
to list a few. Note that an overwhelming majority of these works necessarily lose some (small) approximation factor through compression, and exact solutions are rare, but exist (e.g.~\cite{AhnGM12}).

There is also an indirect connection between our work and lower bounds for distributed graph algorithms (e.g. \cite{SarmaHKKNPPW12,DruckerKO14}),
since our algorithms can be used to obtain upper bounds on the two-party communication complexity of min cuts in some models.\footnote{In a model where Alice and Bob can jointly compute a cut query in $O(\log n)$ communication, and have shared randomness, \pref{thm:gmc} (\pref{thm:stc}) provides a randomized protocol with $\tilde{O}(n)$ (resp. $\tilde{O}(n^{5/3})$) communication for computing the global (resp. $s$-$t$) min cut.}

\subsection*{Organization}
In \pref{sec:minc1}, we present our simple algorithm for global min cut, as well as important algorithmic primitives (such as subsampling edges).
Then in \pref{sec:bksamp}, we introduce our query-efficient implementation of Bencz\'ur and Karger's edge-strength based sampling, after which we demonstrate its application to global min cut in \pref{sec:minc2}.
Finally, \pref{sec:st} contains our result for min $s$-$t$ cuts.

\subsection*{Discussion and Future Work}
The main take-home message of our work is simple, randomized algorithms for \emph{exact} global and $s$-$t$ min cut with $\tilde{O}(n)$ and $\tilde{O}(n^{5/3})$ queries, respectively. In particular, our algorithm for global min cut learns (up to the polylog factors) just enough information to even specify one of the $2^n$ distinct cuts, and both are well below $\tilde{\Theta}(n^2)$. So in this natural oracle model, it is possible to find the exact global and $s$-$t$ cut \emph{without learning the underlying graph}.

Our work also motivates numerous directions for future work: Are weighted or directed graph cuts computable in $o(n^2)$ queries? Do deterministic min cut algorithms exist with truly subquadratic queries? Or can graph cuts still provide a $\Omega(n^2)$ submodular-function-minimization lower bound (perhaps for deterministic algorithms)? While graph cuts are indeed a very special case of submodular functions, can any of the ideas from our work be used in randomized algorithms for a broader class of submodular function minimization?

\section{Global min-cut in $\tO(n)$ queries}
% !TeX root = ..\writeup.tex
%
\label{sec:minc1}
We begin by observing that if $G$ has $m$ edges, we can learn $G$ entirely with $\tO(m)$ queries.
This is because locating a single edge takes only $O(\log n)$ queries.

\begin{lemma}[Learning an edge with $O(\log n)$ queries]\label{lem:1edge}
We can learn one neighbor of a vertex $v \in V$ in $O(\log n)$ queries.
\end{lemma}
\begin{proof}
    To find one neighbor of $v$, we perform the following recursive procedure: we partition $V \setminus v$ into two sets $S_1$ and $S_2$ of sizes $\lfloor\frac{n-1}{2}\rfloor,\lceil\frac{n-1}{2}\rceil$, respectively.
    We then query the cut values of $\{v\}$, $S_i$, and $S_i \cup \{v\}$, from which we can infer how many neighbors $v$ has in $S_i$, for $i \in \{1,2\}$ ($\Big(c(\{v\}) + c(S_i) - c(S_i \cup \{v\})\Big)/2$).
    If $v$ has no neighbors, return ``no neighbors''.
    Otherwise, if $v$ has a neighbor in $S_1$, then proceed recursively in $S_1$; otherwise proceed recursively in $S_2$.
\end{proof}

The above observation suffices to learn the entire graph with $\tO(m)$ queries, as it is easy to modify the algorithm to ignore known neighbors of $v$ (if $S_1$ or $S_2$ contain only known neighbors of $v$, ignore them).
If $m = \tO(n)$, Theorem~\ref{thm:gmc} follows easily.

Otherwise, if $m \gg n$, a natural idea is to randomly subsample the edges of $G$ until we are left with a sparse graph, and use this sparse graph to learn useful data about $G$. Indeed, we show that after a preprocessing step of $n$ queries, sampling a unifomrly random edge only requires $O(\log n)$ queries:

\begin{corollary}[Sampling a uniformly random edge with $O(\log n)$ queries]
    Given oracle access to the cut values of a graph on $n$ vertices, after performing $n$ initial queries we can sample a random edge in $O(\log n)$ additional queries (i.e. $k$ uniformly random edges can be drawn in $n + O(k \log n)$ queries).
\end{corollary}
\begin{proof}
First, as a preprocessing step, we perform $n$ queries to determine the degree of every vertex. 
    Now, we choose a random edge by choosing a random vertex $v$ with probability proportional to its degree, then performing the procedure detailed in the proof of \pref{lem:1edge}, but choosing to recurse on either $S_1$ or $S_2$ randomly with probability proportional to the degree of $v$ into each set.
\end{proof}

Because we can sample random edges, we might hope to subsample $G$ and obtain a sparse graph $G'$ which has the same approximate cut values as $G$. In particular, if the minimum cut of $G$ has value $c$, and we sample each edge independently with probability $\log n /c$, then the subsampled graph $G'$ preserves all cuts with high probability within a $(\log n / c)(1\pm \epsilon)$ factor. However, sampling with probability much smaller than $\log n /c$ will yield poor cut concentration in $G'$.

So if $c \approx \frac{m}{n}$, the next step in our algorithm is to do this simple uniform subsampling and work with $G'$, which will have $\approx n \log n$ edges in expectation. But if $c \ll \frac{m}{n}$, the resulting $G'$ will still have too many edges to learn, so we need some additional work.

Fortunately, when the minimum cut size $c$ is small compared to the average degree, we can preprocess $G$ to an intermediate $G^*$ whose average degree is $\approx c$ without destroying the minimum cut via random contractions.
Our preprocessing step essentially runs Karger's Algorithm \cite{Karger93-Kargers_Algorithm} (reproduced here for completeness) for a well-chosen number of steps (not all the way to termination).

\begin{algorithm}[Karger's Algorithm \cite{Karger93-Kargers_Algorithm}]
    \label{alg:karg}{\color{white}.}\\
    {\bf Input:} A graph $G$.
    \begin{compactenum}
	\item For $j = 1,\ldots,n-2$:
    \begin{compactenum}
	\item Sample a random edge of $G$, and contract its two endoint into a single ``super-vertex''.
	\item Retain multi-edges, but remove self-loops.
    \end{compactenum}
    \end{compactenum}
    {\bf Output:} The cut between the two remaining super-vertices, which form a partition of $G$'s vertices into two sets.
\end{algorithm}
In Karger's seminal paper, he proves that this algorithm finds the minimum cut in a graph with probability at least $\frac{1}{n^2}$, yielding a randomized algorithm for minimum cut.

We will not run Karger's algorithm to its completion, but rather only until there are $cn$ total edges remaining in the graph (we can guess $c$ within a factor of $2$ at the cost of $\log n$ additional iterations).
The following simple lemma shows that with constant probability, the minimum cut will survive:

\begin{lemma}[Karger's Algorithm on small cuts]\torestate{
	\label{lem:karger3}
    Let $G$ be a graph with minimum cut value $c\geq 1$.
    If we run Karger's algorithm on $G$ until there are at most $cn$ edges in the graph, then the minimum cut survives with constant probability.
}
\end{lemma}
We will prove \pref{lem:karger3} in \pref{sec:karg}.
Of course, we must verify that we can run $T$ steps of Karger's algorithm with $\tO(T)$ oracle queries:
\begin{proposition}\label{prop:koracle}
    Given oracle access to the cut values of $G$, we can run $T$ steps of Karger's algorithm using $\tO(T)$ queries.
\end{proposition}
\begin{proof}
    The key observation is that keeping track of super-vertices requires no additional queries---it is simply a matter of treating all vertices belonging to a super-vertex as a single entity.

    Each step of Karger's algorithm requires sampling a random edge, which we have already seen requires $O(\log n)$ oracle queries assuming the degree of every vertex is known.
    In order to keep track of the degree of super-vertices, we require only a single oracle query after every edge contraction: we ask for the cut value between the super-vertex and the remainder of the graph.
\end{proof}

At this point, our algorithm is as follows: we first run Karger's algorithm until the min cut size is comparable to the average degree, then subsample the graph to obtain a sparse graph that approximates the cuts of the original graph well.
Applying concentration arguments, it's easy to see that this algorithm immediately yields an approximate min cut.

However, the following observation allows us to improve upon this, and learn the minimum cut exactly!
Since the cuts in $G'$ approximate the cuts in $G$ well, any two nodes that are together in \emph{every} approximate minimum cut in $G'$ are safe to contract into a super-node (because they certainly aren't separated by the min cut).
After these contractions, if there are sufficiently few edges remaining between the super-vertices, we can learn the entire remaining graph between the super-vertices and find the true minimum cut.

The following structural result shows that this is indeed the case: the total number of edges that participate in non-singleton approximately-minimum cuts is at most $O(n)$.\footnote{A cut is non-singleton if each side has at least two nodes.}

\begin{lemma}[Covering approximate min cuts with $O(n)$ edges]\torestate{\label{lem:cover}
Let $G = (V,E)$ be an unweighted graph with minimum degree $d$ and minimum cut value $c$.
Let $\cC$ be the set of all non-singleton approximate-minimum cuts in the graph, with cut value at most $c+\epsilon d$, for $\epsilon < 1$.
    Then $|\cup_{C \in \cC} C|$ (the total number of edges that participate in cuts in $\cC$) is $O(n)$.}
\end{lemma}
We remark that a similar claim is proven in \cite{KT15-deterministic_min_cut}.
While the theorem of \cite{KT15-deterministic_min_cut} would be sufficient for our purposes,\footnote{Their result is stronger in the sense that they also show how to locate the cover in deterministic time $O(m)$, while our result is slightly shaper since we only require $O(n)$ edges rather than $\tO(n)$.}
 our proof is extremely simple, and so we include it in \pref{sec:coverlem}.

This concludes our global min-cut algorithm.
Below, we summarize the algorithm, and formally prove that it is correct.

\begin{algorithm}[Global Min Cut with $\tO(n)$ oracle queries]
    \label{alg:mincut}{\color{white}.}\\
    {\bf Input:} Oracle access to the cut values of an unweighted simple graph $G$.
    \begin{compactenum}
	\item Compute all of the single-vertex cuts.\label{step:single}
	\item For $c = 2^j$ for $j = 0,1,\ldots,\log n$,\label{step:guess}
    \begin{compactenum}
	\item Repeat $\log n$ times:\label{step:samploop}
	    \begin{compactenum}
		\item Run Karger's Algorithm until there are a total of at most $cn$ edges between the components in the graph, call the resulting graph $G_1$.\label{step:karg}
		\item Starting with $G_1$, subsample each edge with probability $p = \frac{80\ln n}{\epsilon^2 c}$ to obtain a graph $G_2$ (any $\epsilon \in (0,1/3)$ suffices) \label{step:subsample}
		\item Find all non-singleton cuts of size at most $(1+3\epsilon)pc$ in $G_2$, and contract any two nodes which are together in all such cuts, call the resulting graph $G_3$\label{step:contract}
		\item Learn all of $G$'s edges between the super-vertices of $G_3$ to obtain $G_4$ (unless there are more than $n\log n$ edges, in which case abort and return to \pref{step:samploop}).\label{step:learnall}
	\item Compute the minimum cut in $G_4$, and if it is the best seen so far, keep track of it.\label{step:comp}
	    \end{compactenum}
    \end{compactenum}
    \end{compactenum}
    {\bf Output:} Return the best cut seen over the course of the algorithm.
\end{algorithm}

\begin{theorem}[Mincut]
    \pref{alg:mincut} uses $\tO(n)$ queries and finds the exact minimum cut in $G$ with high probability.
\end{theorem}
\begin{proof}
First, we will prove the correctness of the algorithm.
    Clearly, if one of the single-vertex cuts is the minimum cut, the algorithm finds this cut in \pref{step:single}, so suppose that the best cut has value $\hat c < d_{\min}$, where $d_{\min}$ is the minimum degree in $G$.

In one of the iterations of \pref{step:guess}, $c$ is within a factor of $2$ of $\hat c$, and we focus on this iteration.
    In \pref{step:karg}, by \pref{lem:karger3}, the minimum cut survives with at least constant probability.
    By the concentration arguments given in \pref{lem:sample} and \pref{cor:sample}, in \pref{step:subsample} every cut in $G_2$ is close to the value of the cut in $G_1$ with high probability,\footnote{where ``close'' means that the value of the cut in $G_2$ is within $(1\pm \varepsilon)pk$, where $k$ is the value of the cut in $G_1$.} and so no edge in the minimum cut is contracted in \pref{step:contract}.
    Therefore, with constant probability, we find the minimum cut in \pref{step:comp}.
    Since we repeat this process $\log n$ times in \pref{step:samploop}, the total probability that we miss the global min cut in every iteration is polynomially small. This proves the corectness of the algorithm.

    Now, we argue that at most $\tO(n)$ queries are required.
    At every iteration of the inner loop, we run Karger's Algorithm for less than $n$ steps ($\tO(n)$ queries by \pref{prop:koracle}). 
    Then, we subsample each of $cn$ edges each with probability $O(\frac{\ln n}{c})$; with high probability this is equivalent to sampling $\tO(n)$ random edges ($\tO(n)$ queries by \pref{cor:sample}).
    Step \ref{step:contract} does not require any queries.
    By \pref{lem:cover}, the true minimum cut hast only $O(n)$ edges.
    Therefore \pref{step:learnall} requires learning only $O(n)$ edges ($\tO(n)$ queries) if $c$ is the true value of the minimum cut, and otherwise the step is aborted.
    Finally, \pref{step:comp} requires no additional queries.
    Since the inner loop is repeated a total of $\log^2 n$ times, this concludes the proof.
\end{proof}

In the following subsections, we provide proofs of key intermediate lemmas.

% !TeX root = ..\writeup.tex

\subsection{Compressing the graph with Karger's algorithm}
\label{sec:karg}

\restatelemma{lem:karger3}

\begin{proof}
    Fix a specific min cut $C$. We apply Karger's algorithm until the total number of edges drops to $cn$.
    At each step, the probability that we contract an edge from $C$ is at most $1/n$, and we have at most $n$ steps, so the probability that $C$ survives is at least $(1-1/n)^n > 1/4$ for $n > 2$.
\end{proof}

% !TeX root = ..\writeup.tex

\subsection{Subsampling the graph}
First, we show that if we sample with probability proportional to $\tO(1/c)$, every cut in the subsampled graph has value close to its expectation. Because there are $2^n$ cuts, a simple Chernoff bound followed by a union bound is insufficient. Instead, we perform a more careful union bound by appealing to a polynomial bound on the number of approximately minimum cuts (as is standard in this setting, see e.g. \cite{Karger99-skeletons}).

\begin{lemma}\label{lem:sample}
    Let $G = (V,E)$ be a multigraph with minimum cut value $c$, and let $G' = (V,E')$ be the result of sampling each edge of $E$ with probability $p \ge \min\left(\frac{40\ln n}{\epsilon^2 c}, 1\right)$.
    Then with high probability, every cut of value $k$ in $G$ has value $(1\pm \epsilon)pk$ in $G'$.
\end{lemma}
\begin{proof}
For each edge $e \in E$, consider the random binary variable
$X_e \triangleq \begin{cases} 1 & e \in E' \\ 0 & \text{otherwise} \end{cases}$.
    Notice that $\E(X_e) = p$.  Let $C$ be a cut of size $k$.
    By a Chernoff bound, the probability that $C$ has cut value deviating from its expectation by more than an $\epsilon$-factor in $G'$ is bounded by:
\begin{gather}\label{eq:bennett1}
\Pr \left[ \left|\sum_{e\in C^*} X_e - pk\right| > p\epsilon k \right]
    \le 2\exp\left(-\frac{\epsilon^2 pk}{2}\right)
    \le 2 n^{-10k/c},
\end{gather}
where the last inequality follows by our choice of $p$.

Now, it follows from the analysis of Karger's algorithm (\pref{lem:karger} below) that for every integer $\ell > 0$ there are at most $(2n)^{2\ell}$ cuts of value at most $\ell c$.
    Consider a cut $C$ with value in $[\ell c, (\ell+1) c]$ in $G$.
    Using \pref{eq:bennett1}, we have that the probability that its value in $G'$ deviates from expectation by more than $\pm p(\ell + 1) \epsilon c$ is at most $n^{-10 \ell}$.
    Taking a union bound over all such $C$ and all values of $\ell$ the soundness holds with probability at least $1-n^{-6}$.
\end{proof}

The following lemma, which we employed in order to bound the number of cuts of each size, is an oft-used consequence of Karger's algorithm (see e.g. \cite{KS96-Karger_Stein}).
\begin{lemma}[Bound on the number of small cuts]\label{lem:karger}
If a graph on $n$ vertices has a minimum cut of size $c$, then there are at most $(2n)^{2\ell}$ cuts of size $\ell c$.
\end{lemma}
\begin{proof}
Fix a specific cut $C$, such that $|C| = \ell c$.
Consider Karger's algorithm, in which we contract a uniformly random edge in each step.
    After $t$ steps, there are at least $(n-t)c/2$ edges in the graph (since no vertex can ever have degree less than $c$ in Karger's algorithm).
    Then in the $t$-th step of Karger's algorithm there is probability at most $\frac{2\ell c}{(n-t)c}$ that an edge from $C$ is contracted.
    Using a telescoping product argument, the probability that $C$ survives for $n-2\ell$ steps of the algorithm is at least $\prod_{t=0}^{n-2\ell} \left(1-\frac{2\ell}{n-t} \right) = \frac{(2\ell)!}{n(n-1)\cdots (n-2\ell+1)} \ge n^{-2\ell}$.
    After $n-2\ell$ steps, there are $2\ell$ vertices remaining, so less than $2^{2\ell}$ cuts survived.
    Therefore in total there can only be $(2n)^{2\ell}$ such cuts in the original graph.
\end{proof}

As a corollary of \pref{lem:sample}, we have that the approximately minimum cuts of the subsampled graph correspond to approximately minimum cuts in the original graph:

\begin{corollary}\label{cor:sample}
    Let $G = (V,E)$ be a graph with minimum minimum cut value $c$.
    Let $G' = (V,E')$ be the result of sampling each edge in $E$ with probability $p = \min\left(\frac{40 \ln n}{\epsilon^2 c},1\right)$, with $\eps \le 1/3$.
Then the following events occur with high probability:
\begin{description}
    \item[Completeness] the minimum cut of $G$ has value at most $p(1+\epsilon)c$ in $G'$.
    \item[Soundness] every cut of value at most  $p(1+\epsilon)c$ in $G'$ has value at most  $(1+3\epsilon)c$ in $G$.
	Furthermore, no cut has value less than $p(1-\epsilon)c$ in $G'$.
\end{description}
\end{corollary}
\begin{proof}
    This follows immediately from \pref{lem:sample}, because with high probability, every cut concentrates to within a $(1\pm \epsilon)$ factor of its expectation.
\end{proof}

% !TeX root = ..\writeup.tex

\subsection{Covering approximate min cuts with $O(n)$ edges}\label{sec:coverlem}

\restatelemma{lem:cover}

\begin{proof}
For any $\epsilon > 0$, let $0 < \alpha + \beta < \frac{1}{2}(1-\epsilon)$.

Notice that any subset of $\cC$ induces a partition over $V$, where two vertices are in the same component if they are on the same side of every cut.
Let $K$ be a subset of $\cC$ chosen as follows: starting from $K$ empty, while there exists a cut $C \in \cC$ such that adding $C$ to $K$ splits at least one component into two components each of size $\ge \beta d$, add $C$ to $K$.
Suppose that at termination, $K$ contains $k$ cuts $K = \{C_1,\ldots,C_k\}$.
We claim that
\begin{claim}
The union of cuts in $K$ contains at most $(c + \eps d)k$ edges, $|\cup_{i \in [k]} C_i| \le (c+\eps d)k$.
\end{claim}
\begin{proof}
Each of the $k$ cuts is an approximate min cut and contains at most $c + \eps d$ edges.
\end{proof}
\begin{claim}
The number of cuts in $K$ is at most $k \le \frac{n}{\beta d} - 1$.
\end{claim}
\begin{proof}
We say that a component of the partition $K$ is {\em directly charged} at time $t$ if adding $C_t$ splits component into two components of size $\ge\frac{n}{\beta d}$. 
We say that a component of the partition at time $t$ is {\em indirectly charged} at time $t' > t$ if one of its descendent components is directly charged at time $t'$.

We prove by induction that no set of size $x$ is charged more than $\frac{x}{\beta d} -1 $ times.
For the base case, we take $x = 2\beta d$. 
A component of this size can be charged at most once, because it can be charged directly if it is split into two sets of size exactly $\beta d$, but never charged again.
Now, suppose that the claim holds for any $y<x$, and consider a set $X$ of size $x$.
The first time $X$ is split, it becomes refined into two sets $U,W$ of sizes $u$ and $w$.
The set $X$ may be charged during this split.
By induction, $U$ may be charged at most $\frac{u}{\beta d}-1$ times, and $W$ may be charged at most $\frac{w}{\beta d}-1$ times.
So $X$ can be charged at most $(\frac{w}{\beta d} -1) + (\frac{u}{\beta d}-1) + 1 = \frac{x}{\beta d} -1$ times, as desired.

To finish the proof, we note that the set of all nodes is charged at most $\frac{n}{\beta d}-1$ times, and every cut $C_i$ causes the set of all nodes to be charged.
\end{proof}

Now, define $S$ to be the set of vertices $v \in V$ with at least $\alpha \deg(v)$ incident edges in $K$.
We claim that $S$ cannot be too large:
\begin{claim}
The volume of $S$ is bounded, $\sum_{v \in S} \deg(v) \le \frac{4dk}{\alpha}$.
\end{claim} 
\begin{proof}
There are at most $(c + \eps d)k \le 2 d k$ edges in $K$.
Therefore, there are at most $4dk$ edge, vertex pairs for which the edge is in $K$ and is incident to some vertex.
By definition, any vertex $v \in S$ participates in at least $\alpha \deg(v)$ such pairs, which implies that $\alpha \sum_{v \in S} \deg(v) \le 4dk$.
\end{proof}

We will say that a node $v$ is {\em small} for a cut $C$ if adding $C$ to $K$ causes $v$'s connected component to shrink to size $< \beta d$. 
We say $v$ is {\em small} for $K$ if it is small for any cut $C$ of size at most $c + \eps d$.

\begin{claim}
If $v$ is small for $K$ then $v \in S$.
\end{claim}
\begin{proof}
Assume for contradiction that $v \not\in S$, and let $C$ denote the cut for which $v$ is small.
Let $B$ be the new connected component of $v$ if we were to add $C$ to $K$.
Then $v$ has fewer than $\beta d$ neighbors in $B$, because $v$ is small and therefore $|B| < \beta d$.
Since $v \not\in S$, $v$ has at most $\alpha \deg(v)$ edges incident in $K$.
Therefore, $v$ has at least $\deg(v) - \alpha \deg(v) - \beta d$ incident edges crossing the cut $C$.

Since $C$ is not a singleton cut, we can move $v$ to the other side of $C$, and decrease its size by 
\[
\deg(v) - 2\alpha \deg(v) - 2\beta d \ge (1 - 2\alpha -2\beta)\deg(v) > \eps d,
\]
since $2(\alpha + \beta) < 1 -\eps$.
This is a contradiction, since we cannot have a cut of size $< c$.
\end{proof}

From this claim, we have that every edge in $\cup_{C \in \cC} C$ is either in $K$ or incident to a vertex in $S$.
To see why, consider any edge $(u,v) \in C$ which not in $K$.
Since $(u,v)$ is not in $K$, $(u,v)$ are in the same component. 
Also, since we did not add $C$ to $K$, it must be the case that one of $u$, $v$ is small for $C$ and therefore small for $K$ and therefore included in $S$.
We use this to upper bound $|\cup_{C \in \cC} C|$, counting all edges in $K$ and all edges incident on $S$.
We have that
\begin{align*}
|\cup_{C \in \cC}| 
&\le |\cup_{C \in K} C| + \sum_{v \in S} \deg(v)\\
&\le 2dk + \frac{4dk}{\alpha}\\
&\le \left(2d\cdot \frac{n}{\beta d} + \frac{4d}{\alpha}\cdot \frac{n}{\beta d}\right) = O(n).
\end{align*}
where we have first applied our bounds on the volume of $S$ and the number of edges in $K$, then applied our bound on $k$.
The conclusion follows.
\end{proof}

\section{Connectivity-preserving sampling in the oracle model}
% !TeX root = ..\writeup.tex

\label{sec:bksamp}
Now we show how to subsample a graph with arbitrary connectivity to obtain a sparse graph in which all cut values are well-approximated (also known as a \emph{sparsifier}).
The algorithm and analysis are inspired by \cite{BK15-strength_sampling}, but we must make modifications to both in order to optimize query efficiency.
We begin with some definitions.

\begin{definition}
    A graph $G$ is \emph{$k$-strongly-connected} if there is no cut of size less than $k$ in $G$.
    The \emph{strong connectivity} of $G$, denoted $K(G)$, is the size of $G$'s minimum cut.
\end{definition}

\begin{definition}
    Given a graph $G = (V,E)$ and an edge $e = (u,v) \in E$, define $e$'s \emph{strength} $k_e$ to be the maximum of the strong connectivities over all vertex-induced subgraphs of $G$ containing $e$:
    \[
	k_e = \max_{S \subseteq V~:~u,v \in S} K(G[S]),
    \]
where $G[S]$ denotes the vertex-induced subgraph of $G$ on $S$.
\end{definition}

The following theorem, due to Bencz\'{u}r and Karger, shows that if we sample each edge with probability inverseley proportional to its strength, every cut will be well-preserved.
\begin{theorem}[Bencz\'{u}r and Karger \cite{BK15-strength_sampling}]\label{thm:BK}
Let $G=(V,E)$ be an unweighted graph.
For each edge $e \in E$, let $k_e$ denote the edge strength of $e$.
Suppose we are given $\{k'_e\}_{e \in E}$ such that $\tfrac{1}{4}k_e \le k'_e \le k_e$.
Let $H$ be the graph formed by sampling each edge $e$ with probability
\[
p_e = \min\left(\frac{100\ln n}{k'_e \epsilon^2},1\right),
\]
and then including it with weight $1/p_e$.
Then with high probability, $H$ has $O(n\ln n/\epsilon^2)$ edges, and every cut in $H$ has value $(1\pm \epsilon)$ of the original value in $G$.
\end{theorem}

While Bencz\'ur and Karger give efficient algorithms for computing approximate edge strengths when the graph is known, in our setting we cannot afford to look at every edge.
The following algorithm shows how to compute approximate edge strengths, and how to compute the sparsifier $H$, with $\tO(n/\epsilon^2)$ oracle queries.

\begin{algorithm}[Approximating Edge Strengths (and sampling a sparsifier $H$)]{\color{white}.\label{alg:edgestrngth}}\\
{\bf Input:} An accuracy parameter $\epsilon$, and a cut-query oracle for graph $G$.
    \begin{compactenum}
	\item (Initialize an empty graph $H$ on $n$ vertices).
	\item For $j = 0,\ldots,\log n$, set $\kappa_j = n2^{-j}$ and:
\begin{compactenum}
\item Subsample $G'$ from $G$ by taking each edge of $G$ with probability $q_j = \min(100 \cdot 40 \cdot \frac{\ln n}{\kappa_j},1)$\label{step:sub1}
\item In each connected component of $G'$:
\begin{compactenum}
\item While there exists a cut of size $\le q_j \cdot \frac{4}{5}\kappa_j$, remove the edges from that cut, and then recurse on the two sides.
Let the connected components induced by removing the cut edges be $C_1,\ldots,C_r$.\label{step:minc}
\item For every $i \in [r]$ and every edge (known or unknown) with both endpoints in $C_i$, set the approximate edge strength $k'_e := \frac{1}{2}\kappa_j$ (alternatively, subsample every edge in $C_i \times C_i$ with probability $2q_j/\epsilon^2$ and add it to $H$ with weight $\epsilon^2/2q_j$).\label{step:Hsamp}
\item Update $G$ by contracting $C_i$ for each $i \in [r]$.
\end{compactenum}
\end{compactenum}
    \end{compactenum}
{\bf Output:} The edge strength approximators $\{k'_e\}_{e \in E}$ (or the sparsifier $H$).
\end{algorithm}

\begin{theorem}\label{thm:qbk}
For each edge $e \in G$, the approximate edge strength given in \pref{alg:edgestrngth} is close to the true edge strength,
$\frac{1}{4}k_e \le k_e' \le k_e$.
Furthermore, the algorithm requires $\tO(n/\epsilon^2)$ oracle queries to produce the sparsifier $H$, which satisfies:
    \begin{compactitem}
    \item $H$ has $O(n\ln n/\epsilon^2)$ edges
    \item The maximum weight of any edge $e$ in $H$ will be $O(\epsilon^2 k_e/\ln n)$
    \item Every cut in $H$ is within a $(1\pm \epsilon)$-factor of its value in $G$.
    \end{compactitem}
\end{theorem}
\begin{proof}
The proof follows from two claims, which we state here and prove later:
\begin{claim}\label{claim:complete}
    At iteration $j =\lceil \log (n/k_e)\rceil$, the edge $e$ is either assigned $k_e' =  \frac{1}{2}\kappa_j = n/2^{j+1} \ge k_e/4$ or has already been assigned a larger value of $k_e'$.
\end{claim}
\begin{claim}\label{claim:sound}
    At iteration $j$, no edges $e$ with $k_e < \frac{1}{2}\kappa_j$ are assigned a strength approximation.
\end{claim}

Given these two claims, we have that the approximate edge strength of every edge is within a factor of two of the true strength.
Furthermore, to construct $H$, each iteration only requires $\tO(n/\eps^2)$ cut queries.
    In \pref{step:sub1}, all components with strong connectivity larger than the current connectivity ($\kappa_j$) have been contracted, so there are no $2\kappa_j$-connected components. By \pref{cor:kcon} (stated shortly), the current $G$ therefore has at most $O(n\kappa_j)$ edges.
    Therefore, in \pref{step:sub1} we have $q_j = \tO(n/|E|)$, and the expected number of sampled edges is therefore just $\tO(n)$, and this step requires only $\tO(n)$ cut queries.
The operations in \pref{step:minc} require no additional queries.
Finally, again by \pref{cor:kcon}, \pref{step:Hsamp} requires at most $\tO(n/\epsilon^2)$ queries, and the consequent step requires no samples.
    The whole process is iterated $O(\log n)$ times, for a total of $\tO(n/\epsilon^2)$ queries.
    The listed properties of $H$ follow from \pref{thm:BK}.

Now, we prove our initial claims.

\medskip

To prove \pref{claim:complete}, consider the strongly connected component of strength $k_e$ that $e$ belongs to, $C_e$.
    Since we subsample edges with probability $q_j = 100 \cdot 40 \cdot \ln n/\kappa_j \ge 100 \cdot 40 \ln n/k_e$, with high probability every cut of $C_e$ has size at least $\frac{9}{10} q_j k_e \ge \frac{9}{10} q_j \kappa_j$in $G'$ (by concentration bounds identical to those in \pref{lem:sample}).
    Therefore, no minimum cut removed in \pref{step:minc} will disconnect $C_e$.
The claim follows.

\medskip

    To prove \pref{claim:sound}, we note that by definition if $k_e < n/2^{j+1}$, then $e$ cannot participate in any vertex-induced component with strong connectivity $\kappa_{j}/2$.
    We will prove that every component $C_1,\ldots,C_r$ created in \pref{step:minc} is at least $(\kappa_{j}/2)$-connected.
    For this, it is necessary to prove that any cut of size less than $\kappa_{j}/2$ is removed.
Let $C = \cup C_i$ be the components of $G'$ after \pref{step:sub1}.
First, we notice that at most $n$ cuts in $C$ are necessary to remove all non-strongly-connected edges.
    Let $S_1,\ldots,S_\ell$ be a sequence of at most $\ell \le n$ cuts with sizes $a_1,\ldots,a_\ell$ respectively, so that $a_i \le \kappa_{j}/2$ in $G$ when restricted to the vertex-induced subgraph given by the vertices of $C$.
Let $a_1',\ldots,a_\ell'$ be the sizes of the cuts $S_1,\ldots,S_\ell$ in C (in the subsampled graph $G'$).

By a Chernoff bound,
\[
\Pr[ a_i' - q_j a_i \ge s \cdot q_j a_i]
\le \begin{cases}
\exp\left(-s q_j a_i/3 \right) & s \ge 1\\
\exp\left(-s^2 q_j a_i /3\right) & s \le 1
\end{cases}
\]
We choose $s = \frac{4}{5} \frac{\kappa_{j}}{a_i } - 1$ so that $(1+s)q_j a_i = q_j\cdot \frac{4}{5} \kappa_{j}$.
Then because $a_i \le \kappa_{j}/2$,
\[
    sq_j a_i = q_j\cdot \frac{4}{5} \cdot \kappa_{j} - q_j a_i \ge q_j \cdot \frac{3}{10} \cdot \kappa_{j} \ge 30\ln n,
\]
and and because $a_i \le \kappa_{j}/2$, $s \ge \frac{3}{5}$, so
\[
s^2 q_j a_i \ge 18\ln n.
\]
Thus, the probability that any of the cuts $S_i$ has size $a_i' \ge \frac{4}{5}q_j \kappa_j$ in the subsampled graph $G'$ is at most $n^{-6}$.
Taking a union bound over all of the $S_i$, we have that with high probability, all of the $S_i$ will be small enough in the subsampled graph to be removed.
\end{proof}

To argue that we did not sample too many edges (or require too many oracle queries) in \pref{step:sub1}, we must bound the number of edges with strength at least $k$ and at most $2k$.
The following lemma is the crux of the argument (this lemma is not novel and has appeared elsewhere, e.g. \cite{BK15-strength_sampling}).
\begin{lemma}\label{lem:kcon}
Let $G = (V,E)$ be a weighted graph without self-loops, and let $|V|=n$.
    Denote by $w(E)$ the total weight of the edges in $E$.
    If $w(E) \ge d(n-1)$, then $G$ contains a strongly $d$-connected component.
\end{lemma}
\begin{proof}
The proof is by induction---if $n = 2$, the conclusion is obvious.
Now, by contradiction, let $n$ be the smallest integer for which this is not the case.
    Since $G$ is not $d$-connected, by removing a set of edges of total weight $<d$, we can split $G$ into two components $C_1,C_2$ of size $n_1$ and $n_2$ with edge sets $E_1$ and $E_2$, so that the total weight of edges amongst the two parts is at least $w(E_1) + w(E_2) \ge d(n-2)+ 1$.
    Since $G$ and all of its induced subgraphs have no $d$-strongly-connected subgraphs, by the induction hypothesis both $C_1$ and $C_2$ must have $w(E_1) \le d(n_1 - 1)$ and $w(E_2) \le d(n_2 - 1)$.
But then
    $w(E_1) + w(E_2) \le d(n_1 + n_1 - 2) = d(n-2)$, which is a contradiction.
This completes the proof.
\end{proof}

\begin{corollary}\label{cor:kcon}
In an graph on $n$ vertices which has strong connectivity $k$ and no components with strong connectivity $\ge 2k$, there are $\Theta(nk)$ edges.
\end{corollary}
\begin{proof}
In a strongly $k$-connected component, every vertex must have degree at least $k$, which gives the lower bound.
To see the upper bound, we invoke \pref{lem:kcon} (which gives the desired conclusion by taking $d = 2k$).
\end{proof}

\section{Global min-cut revisited}
% !TeX root = ..\writeup.tex

\label{sec:minc2}

Now that we are in posession of a more sensitive sampling algorithm, we give a simplified global min cut algorithm (``simplified'' by pushing all the complexity to the sampling procedure).

\begin{algorithm}[Simpler global Min Cut with $\tO(n)$ oracle queries]
    \label{alg:mincut2}{\color{white}.}\\
    {\bf Input:} Oracle access to the cut values of an unweighted simple graph $G$.
    \begin{compactenum}
	\item Compute all of the single-vertex cuts.\label{step:single2}
    \item Compute a sparsifier $H$ of $G$ using \pref{alg:edgestrngth} with $G$ and with small constant $\epsilon$.\label{step:bk2}
		\item Find all non-singleton cuts of size at most $(1+3\epsilon)$ times the size of the minimum cut in $H$, and contract any edge which is not in such a cut, call the resulting graph $G'$. \label{step:contract2}
		\item If the number of edges between the super-vertices of $G'$ is $O(n)$, learn all of the edges between the super-vertices of $G'$, and compute the minimum cut.\label{step:comp2}
    \end{compactenum}
    {\bf Output:} Return the best cut seen over the course of the algorithm.
\end{algorithm}

\begin{theorem}
    \pref{alg:mincut2} uses $\tO(n)$ queries and finds the exact minimum cut in $G$ with high probability.
\end{theorem}
\begin{proof}
Let $C^*$ be a minimum cut in $G$, and suppose the size of $C^*$ is $c$.
By \pref{thm:qbk}, the sampling performed in \pref{step:bk2} will ensure that with high probability the minimum cut of $H$ has value at least $(1-\epsilon)c$, and that the size of $C^*$ in $H$ is at most $(1+\epsilon)c$.
For $\epsilon < 1/3$,
\[
\frac{(1+\epsilon)c}{(1-\epsilon)c} = 1 + \frac{2\epsilon}{1-\eps} < 1 + 3 \epsilon\mper
\]
Therefore, in \pref{step:contract2} no edge in $C^*$ will be contracted.
Finally, by \pref{lem:cover} at most $O(n)$ edges are left between the super-vertices of $G'$ in \pref{step:comp2} (whp, assuming that all cuts are indeed preserved within $(1 \pm \epsilon)$).
Therefore, if $C^*$ is a non-singleton cut, it (or a cut of the same size) will be found.
No step requires more than $\tO(n)$ queries.
\end{proof}

\section{$s-t$ min-cut in $\tO(n^{5/3})$ queries}
% !TeX root = ..\writeup.tex
\label{sec:st}

Now, we use the low-query sampling algorithm developed in \pref{sec:bksamp} to obtain sub-quadratic query complexity for computing min $s$-$t$ cuts in undirected and unweighted graphs.
Our algorithm follows the same general strategy as the minimum cut algorithm from the previous section: sample a connectivity-preserving weighted graph from $G$, then compress the graph by contracting edges that do not participate in the minimum cut.

\begin{algorithm} [$s$-$t$ min cut with $\tO(n^{5/3})$ queries]\label{alg:st-cut}{\color{white}.}\\
    {\bf Input:} Oracle access to the cut values of an unweighted simple graph $G$.
\begin{compactenum}
    \item Compute a sparsifier $H$ of $G$ using \pref{alg:edgestrngth} with $G$ and with $\epsilon = n^{-1/3}$.\label{step:bksamp}
    \item Compute a maximum $s$-$t$ flow in $H$, and remove the participating edges from $H$; denote the result $H'$\label{step:subflow}
	\item Obtain $G'$ from $G$ by contracting all components that are $3 \epsilon\cdot c$-connected in $H'$.
	\item Learn all edges of $G'$ and compute the minimum $s$-$t$ cut in the resulting graph.\label{step:getcut}
\end{compactenum}
    {\bf Output:} The minimum $s$-$t$ cut computed in \pref{step:getcut}.
\end{algorithm}

\begin{theorem}
    \pref{alg:st-cut} finds an exact $s$-$t$ minimum cut in $\tilde O(n^{5/3})$ oracle calls.
\end{theorem}

\begin{proof}
    Our proof is based on the following claim:
    \begin{claim}
	The number of edges between the super-vertices in $G'$ is at most $O(n^{5/3})$.
    \end{claim}
	 The sparsifier $H$ output in \pref{step:bksamp} by \pref{alg:edgestrngth} has $O(n\ln n/\epsilon^2)$ edges and preserves all cuts to within a multiplicative $(1\pm \epsilon)$.
    In particular the value of the $s$-$t$ maximum flow is at most $n$, and so it is preserved to within an additive $\pm \epsilon n$.

    Then, in \pref{step:subflow} we compute an (exact) $s$-$t$ maximum flow $F$ in $H$, and subtract $F$ from $H$ to obtain the graph $H'$.
    Note that without loss of generality, $F$ is integral and non-circular.
Let $f_H, f_G \leq n$ denote the size of the minimum $s$-$t$-cut in $G,H$ (respectively).
Since each edge has strength at most $n$ in $G$, each edge has weight at most $\epsilon^2 n$ in $H$
    Therefore, by \pref{lem:flow-cover} (stated shortly), the total weight in flow $F$ is at most $O(n\sqrt{f_H \cdot n \epsilon^2})$; since  $f_H \approx f_G \leq n$ (up to a $(1\pm \epsilon)$ factor), this simplifies to $O(\epsilon n^2)$ total weight.

If we could subtract exactly the maximum flow in $G$, we could safely contract all remaining connected components (since the max flow certainly saturates a min $s$-$t$ cut). Since the (exact) $s$-$t$ maximum flow in $H$ approximates the flow in $G$ to within an additive $\pm \epsilon n$ error, we claim that we can safely contract any $3\epsilon n$-connected component in $H'$:

Let $C$ be a $3\epsilon n$-connected component in $H'$.
    Assume by contradiction that there is a minimum $s$-$t$ cut that separates $C$. Because we preserved all cuts to within a multiplicative $(1\pm \epsilon)$, the same cut has value at most $(1+ \epsilon) f_G \leq f_H + 2\epsilon n$ in $H$. But because there is an $s$-$t$ flow of value $f_H$ in $H\setminus H'$, all cuts have value at least $f_H$ in $H \setminus H'$. Therefore, this approximate min $s$-$t$ cut \emph{must} cut at most $2\epsilon n$ edges in $H'$. So immediately by definition of $k$-connectivity, we obtain a contradiction to this cut possibly separating a $3\epsilon n$-connected component in $H'$. This establishes the correctness of the algorithm, since the exact min $s$-$t$ cut is not altered in \pref{step:subflow}.

    Once we contract the $3\epsilon n$-connected components in $H'$, we are left (by \pref{lem:kcon}) with a total weight of at most $3 \epsilon n^{2}$ in $H'$.

	After applying the same contractions to $H$, we have that the total remaining weight is at most $3\epsilon n^2 + O(\epsilon n^2) = O(\epsilon n^2)$ (the sum of the flow and $H'$); and therefore, since the cut around each of the contracted vertices is the same in $G$ and $H$ up to a factor of $(1\pm \epsilon)$, we have that the number of edges remaining in the contracted graph $G'$ is also $|E'| = O(\epsilon n^2)$.

	The total number of queries necessary is $\tO(n/\epsilon^2)$ in \pref{step:bksamp}, and then another $|E'|$ in step \pref{step:getcut}.
	Choosing $\epsilon = n^{-1/3}$ balances the terms, so that we have $|E'|,n/\epsilon^2 \le n^{5/3}$.
	This concludes the proof.
\end{proof}

\subsection{Covering $s$-$t$ min cuts with $O(n^{3/2})$ edges}
\begin{lemma}[Flow cover]\label{lem:flow-cover}
In an undirected graph $G=(V,E)$ with integral weights from $[0,W]$,
every non-circular $s$-$t$ flow (for any $s,t \in V$) of value $f$ uses edges of at most $O(n\sqrt{fW})$ total weight.
\end{lemma}

\begin{proof}
Consider the induced {\em flow graph}, i.e. the DAG that has an edge from $u$ to $v$ with weight equal to the flow from $u$ to $v$.
Fix a topological sorting of the flow graph. We define the {\em length} of an edge to be the difference between its endpoints in the sorting.

Bucket all the edges into $O(\log W)$ buckets according to their weights, with bucket $B_w$ containing all the edges of weight in $[w,2w-1]$. Let $d_w$ denote the (unweighted) average incoming degree when only considering edges from $B_w$, and let $\ell_w$ denote the (unweighted) average length of edges in $B_w$.

For each $i \in [n-1]$, at most $f/w$ edges from $B_w$ cross the cut $C_i$ between the first $i$ vertices and the last $n-i$ vertices in the toplogical ordering (because each such edge has weight $\geq w$ and the total flow crossing any of these cuts is exactly $f$). Similarly, the number of cuts each that edge in $B_w$ crosses is exactly equal to its length. We can count the total number of pairs $(e, C_i)$ such that edge $e\in B_w$ crosses cut $C_i$ in two different ways: summing across $(n-1)$ cuts, or summing across $|B_w|$ edges. We therefore have that
\begin{gather}\label{eq:Vf>Eell}
(n-1) \cdot f/w \geq |B_w| \cdot \ell_w.
\end{gather}

For each vertex $v$, each incoming edge has a different length; therefore, the average length among its incoming edges is at least half of its degree. By the Cauchy-Schwartz inequality it follows that this is also true on average across all edges and vertices:\footnote{To see this, we can compute the average length ($\ell(e)$ denotes the length of edge $e$) as $(\sum_ v\sum_{e \text{ incoming to } v} \ell(e))/\sum_v d_v \geq (\sum_{v} d_v^2/2 )/\sum_v d_v \geq ((\sum_v d_v)^2/2n)/\sum_v d_v \geq \sum_v d_v/2n = d_w/2$.}
\begin{gather}\label{eq:ell>d}
\ell_w \geq d_w/2.
\end{gather}

Observe also that  $|B_w| = n\cdot d_w$.
Combining this observation with Inequalities \eqref{eq:Vf>Eell} and \eqref{eq:ell>d}, we have that
    \begin{align*}
f/w \geq \ell_w \cdot \left(\frac{B_w}{n}\right) \geq d_w^2 /2.
\end{align*}
In particular, for each bucket, the number of edges is bounded by $|B_w| = O(n \sqrt{f/w})$

Therefore, the total weight of all edges is bounded by
\begin{gather*}
\sum_w |B_w| \cdot 2w = \sum_w O(n \sqrt{f w}) = O(n \sqrt{f W}).
\end{gather*}

\ignore{
\medskip

    To see that this is tight, consider the following layered graph: the vertex set $V$ is the union of the sets $\{s\}$, $B_s$,$L_1,\ldots,L_{k}$, $B_t$ and $\{t\}$.
   The edge set $E$ is as follows: $s$ is adjacent to every vertex in $B_s$, $t$ is adjacent to every vertex in $B_t$, for every $0 \le i \le k-1$ there is a complete bipartite graph between $L_i$ and $L_{i+1}$, as well as between $B_s$ and $L_1$ and between $L_k$ and $B_t$.

    We set $|B_t| = |B_s| = n$, set $|L_i| = \sqrt{n}$, and set $k = \sqrt{n}$.
    There are thus a total of $N = 3n + 2$ vertices in the graph.
    Furthermore, it is easy to see that every edge between each $L_i$ and $L_{i+1}$ participates in a min $s$-$t$ cut of value $\binom{\sqrt n}{2}$.
    Therefore, at least $O(N^{3/2})$ edges are needed for the cover.}
\end{proof}

\section*{Acknowledgements}
The authors thank Robert Krauthgamer, Satish Rao, Aaron Schild, and anonymous reviewers for helpful conversations and suggestions.
We also thank Troy Lee for bibliographical pointers and for pointing out errors in a previous version of the proof of \pref{lem:cover}.

\addreferencesection
\bibliographystyle{amsalpha}
\bibliography{mincut}

\end{document}